\documentclass[11pt,a4paper]{article}
\usepackage[]{geometry}
\usepackage{sectsty}
\sectionfont{\normalsize\bfseries}
\subsectionfont{\normalsize\sf}
\usepackage[mathcal]{euscript}
\usepackage{bm,url}                     
\usepackage{amsfonts,todonotes}
\usepackage{amssymb}
\usepackage{amsmath}
\usepackage{amsthm}
\usepackage{graphicx} 
\usepackage{subcaption}              
\usepackage{natbib}                 
\setcitestyle{authoryear,open={(},close={)}}                 
\usepackage{colortbl}               
\usepackage{booktabs}
\usepackage[english]{babel}
\usepackage{amsmath}
\usepackage{amsfonts}
\usepackage{amssymb}
\usepackage{amsthm}
\usepackage{graphicx}
\usepackage{color}
\usepackage{array}
\usepackage{mathrsfs}
\usepackage{setspace}
\usepackage{hyperref}
\definecolor{violet}{rgb}{0.7,0,0.6}

\usepackage{amsfonts}
\usepackage{amssymb}
\usepackage{amsmath}
\usepackage{amsthm}
\newcommand{\E}{\mathcal{E}}

\newtheorem{proposition}{{\sc\bf Proposition}}
\newtheorem{theorem}{{\sc\bf Theorem}}
\newtheorem{lemma}{{\sc\bf Lemma}}

\newtheorem{remark}{{\sc\bf Remark}}
\newtheorem{corollary}{{\sc\bf Corollary}}

\def\E{{\mathbb E}}

\def\I{\mathbb{I}}

\def\cals_+{{\cals_+}}

\def\calh{{\mathcal{H}}}

\def\cals{{\mathcal{S}}}

\newcommand{\var}{{\rm Var}}

\newcommand{\dd}{\mathrm{d}}
\newcommand{\wh}{\widehat}

\newcommand{\std}{\stackrel{\rm d}{\rightarrow}}

\newcommand{\argmin}{\operatornamewithlimits{argmin}}

\newcommand{\operK}{\mathcal{K}}

\allowdisplaybreaks 
\definecolor{violet}{rgb}{0.7,0.2,0.6}

\begin{document}

\begin{center}
\large \bf  On Mahalanobis distance in functional settings
\end{center}
\normalsize

\begin{center}
  Jos\'e R. Berrendero, Beatriz Bueno-Larraz, Antonio Cuevas \\
  Departamento de Matem\'aticas\\
  Universidad Aut\'onoma de Madrid, Spain
\end{center}

\begin{abstract}
\footnotesize {  Mahalanobis distance is a classical tool in multivariate analysis. We suggest here an extension of this concept to the case of functional data. More precisely, the proposed definition concerns those statistical problems where the sample data are real functions defined on a compact interval of the real line. The obvious difficulty for such a functional extension is the non-invertibility of the covariance operator in infinite-dimensional cases. Unlike other recent proposals, our definition is suggested and motivated in terms of the Reproducing Kernel Hilbert Space (RKHS) associated with the stochastic process that generates the data.  The proposed distance is a true metric; it depends  on a unique real smoothing parameter which is fully motivated in RKHS terms. Moreover, it shares some properties of its finite dimensional counterpart: it is invariant under  isometries, it can be consistently estimated from the data and its sampling distribution is known under Gaussian models.  An empirical study for two statistical applications, outliers detection and binary classification, is included. The obtained results are quite competitive when compared to other recent proposals of the literature. }
\end{abstract}


\normalsize


\section{Introduction}\label{sec:intro}


\noindent \textit{The classical (finite-dimensional) Mahalanobis distance and its applications}

Let $X$ be a random variable taking values in ${\mathbb R}^d$ with non-singular covariance matrix $\Sigma$. In many practical situations  it is required to measure  the distance between two points $x_1,x_2\in{\mathbb R}^d$ when considered as two possible observations drawn from $X$. Clearly, the usual (square) Euclidean distance $\Vert x_1-x_2\Vert^2=(x_1-x_2)'(x_1-x_2):=\langle x_1-x_2,x_1-x_2\rangle$ is not a suitable choice  since it disregards the standard deviations and the covariances of the components of $x_i$ (given a column vector $x\in {\mathbb R}^d$ we denote by $x'$ the transpose of $x$). Instead, the most popular alternative is perhaps the classical Mahalanobis 
distance, $M(x_1,x_2)$, defined as 

\begin{equation}\label{eq:MahMult}
M(x_1,x_2) \ = \ \left((x_1-x_2)'\Sigma^{-1}(x_1-x_2)\right)^{1/2}.
\end{equation}

Very often the interest is focused on studying ``how extreme'' a point $x$ is within the distribution of $X$; this is typically evaluated in terms of $M(x,m)$, where $m$ stands for the vector of means of $X$. 

This distance is named after the Indian statistician P. C. Mahalanobis (1893-1972) who first proposed and analyzed this concept \citep{mahalanobis1936} in the setting of Gaussian distributions. Nowadays, some popular applications of the Mahalanobis distance are: supervised classification, outlier detection (\cite{rou90} and \cite{pen96}), multivariate depth measures (\cite{zuo00}), hypothesis testing (through Hotelling's statistic, \citet[Ch. 5]{ren12}) or goodness of fit (\cite{mar75}). This list of references is far from exhaustive.

\

\noindent  \textit{On the difficulties of defining a Mahalanobis-type distance for functional data} 

Our framework here is Functional Data Analysis (FDA); see, e.g., \cite{cuevas2014} for an overview. In other words, we deal with statistical problems involving functional data. Thus our sample is made of trajectories $X_1(t),\ldots,X_n(t)$ in $L^2[0,1]$ drawn from a second order stochastic process $X(t),\ t\in[0,1]$ with $m(t)={\mathbb E}(X(t))$.  The inner product and the norm in $L^2[0,1]$ will be denoted by $\langle\cdot,\cdot\rangle_2$ and $\Vert\cdot\Vert_2$, respectively (or simply $\langle\cdot,\cdot\rangle$ and $\Vert\cdot\Vert$ when there is no risk of confussion). 
We will henceforth assume that the covariance function $K(s,t)=\mbox{Cov}(X(s),X(t))$ is continuous and positive definite. The function $K$ defines a linear operator  $\operK:L^2[0,1]\rightarrow L^2[0,1]$, called covariance operator, given by
\begin{equation}\label{Eq:OpCov}
\operK f(t) \ = \ \int_0^1 K(t,s)f(s)\mathrm{d}s.
\end{equation}

The aim of this paper is to extend the notion of the multivariate (finite-dimensional) Mahalanobis distance \eqref{eq:MahMult} to the functional case when $x_1,x_2\in L^2[0,1]$. Clearly, in view of \eqref{eq:MahMult}, the inverse $\operK^{-1}$ of the functional operator $\operK$ should play some role in this extension if we want to keep a close analogy with the multivariate case. Unfortunately, such a direct approach utterly fails since, typically, $\operK$ is not invertible in general as an operator, in the sense that there is no linear continuous operator $\operK^{-1}$ such that  $\operK^{-1}\operK=\operK\operK^{-1}=  {\mathbb I}$,  the identity operator. 

To see the reason for this crucial difference between the finite and the infinite-dimensional cases, let us recall that some elementary linear algebra yields the following representations for $\Sigma x$ and $\Sigma^{-1}y$,
\begin{equation}\label{eq:inv_finito}
\Sigma x=\sum_{i=1}^d \lambda_i (e_i'x)e_i, \ \  \Sigma^{-1} y=\sum_{i=1}^d \frac{1}{\lambda_i} (e_i'x)e_i
\end{equation}
where $\lambda_1,\ldots, \lambda_d$ are the, strictly positive, eigenvalues of $\Sigma$ and $\{e_i,\ldots,e_d\}$ the corresponding orthonormal basis of eigenvectors. 

In the functional case, the classical Karhunen-Lo\`eve Theorem (see, e.g., \cite{ash2014topics}) provides $X(t)=\sum_jZ_je_j(t)$ (in $L^2$ uniformly on $t$) where the $\{e_j\}$ is the basis of orthonormal eigenfunctions of $\operK$ and the $Z_j=\langle X,e_j\rangle$ are uncorrelated random variables with $\mbox{Var}(Z_j)=\lambda_j$, the  eigenvalue of $\operK$ corresponding to $e_j$.  Then, we have
$$
\operK x=\operK \Big(\sum_{i=1}^\infty \langle x,e_i\rangle e_i\Big)
$$
 Note that the continuity of $K(s,t)$ implies  $\int_0^1\int_0^1K(s,t)^2\mathrm{d}s\mathrm{d}t < \infty$, thus $\operK$ is in fact a compact, Hilbert-Schmidt operator. In addition, it is easy to check $\int_0^1\int_0^1K(s,t)^2\mathrm{d}s\mathrm{d}t=\sum_{i=1}^\infty \lambda_i^2$ so that, in particular, the sequence $\{\lambda_i\}$ converges to zero very quickly. As a consequence, there is no hope of keeping a direct analogy with \eqref{eq:inv_finito}  since
\begin{equation}\label{eq:inv_infinito}
\operK^{-1} x=\sum_{i=1}^\infty \frac{1}{\lambda_i} \langle x,e_i\rangle e_i
\end{equation}
will not define in general a continuous operator with a finite norm. Still, for some particular functions $x=x(t)$  the series in \eqref{eq:inv_infinito} might be convergent. Hence we could use it formally to define the following template which, suitably modified, could lead to a general, valid definition for a Mahalanobis-type distance between two functions $x$ and $m$,
\begin{equation}\label{eq:template}
\widetilde M(x,m)=\left(\sum_{i=1}^\infty \frac{\langle x-m,e_i\rangle^2}{\lambda_i}\right)^{1/2},
\end{equation}
for all $x,m\in L^2[0,1]$ such that the series in \eqref{eq:template} is finite. We are especially concerned with the case where $x$ is a trajectory from a stochastic process $X(t)$ and $m$ is the corresponding mean function.  As we will see below, this entails some especial difficulties. 

\

\noindent \textit{The organization on this work}

 In the next section some theory of RKHS and its connection with the Mahalanobis distance  is introduced, together with the proposed definition. In Section \ref{sec:prop} some properties  of the proposed distance are presented and compared with those of  the original multivariate definition.  Then, a consistent estimator is analyzed  in Section \ref{sec:consist}.  Finally,  some  numerical outputs corresponding to different statistical applications  can be found in Section \ref{sec:appl}.

\section{A new definition of Mahalanobis distance for functional data}\label{sec:def}

Motivated by the previous considerations,  \cite{galeano2015} and \cite{ghi17} have suggested two functional Mahalanobis-type distances,  that we will comment at the end of this section. These proposals are natural extensions to the functional case of the multivariate notion \eqref{eq:MahMult}.   Moreover, as suggested by the practical examples considered in both works, these options performed quite well in many cases. 
However, we believe that there is still some room to further explore the subject for the reasons we will explain  below.

 In this section we will propose a further definition of a Mahalanobis-type distance, denoted $M_\alpha$. Its most relevant features can be summarized as follows:
 \begin{itemize}
 	\item $M_\alpha$ is also inspired in the natural template \eqref{eq:template}. The serious convergence issues appearing in \eqref{eq:template} are  solved by smoothing.
 	\item  $M_\alpha$  depends on a single, real, easy to interpret smoothing parameter $\alpha$ whose choice is not critical, in the sense that the distance has some stability  with respect to $\alpha$.
 	Hence, it is possible to think of a cross-validation or bootstrap-based choice of $\alpha$. In particular, no auxiliary weight function is involved in the definition.
 	\item  $M_\alpha(x,m)$ is a true metric which is defined for any given pair $x, m$ of functions in $L^2[0,1]$.
 	It shares some  invariance properties with the finite-dimensional counterpart \eqref{eq:MahMult}. 
 	\item If $m(t) = {\mathbb E}X(t)$, the distribution of $M_\alpha(X,m)$ is explicitly known  for Gaussian processes. In particular, ${\mathbb E}(M_\alpha^2(X,m))$ and $\mbox{Var}(M_\alpha^2(X,m))$ have explicit, relatively simple expressions. 
 \end{itemize}
 
\

The main contribution of this paper is to show that the theory of Reproducing Kernel Hilbert Spaces (RKHS) provides a natural and useful framework in order to propose an extension of the Mahalanobis distance to the functional setting, satisfying the above mentioned properties. So we next give, for the sake of completeness, a very short overview of the RKHS theory, just focused on the features we will use here.  We refer to \cite{berlinet2004}, Appendix F in \cite{janson1997}  and \cite{scholkopf2002}, for a more detailed treatment of the subject. 

\subsection{RKHS's and the Mahalanobis distance}

The starting element in the construction of an RKHS space of real functions in $[0,1]$ is a positive semidefinite function $K(s,t)$, $s,t\in[0,1]$. For our purposes, $K$ will be the continuous positive definite covariance function of the process $X(t)$ that generates our functional data.

Let us first consider the following auxiliary space $\mathcal{H}_0(K)$ of functions generated by~$K$,
\begin{equation}\label{Eq:DefH0}
\mathcal{H}_0(K) := \{ f:[0,1]\rightarrow {\mathbb R} : \ f(\cdot) = \sum_{i=1}^n a_i K(t_i,\cdot), \ a_i\in\mathbb{R}, \ t_i\in[0,1], \ n\in\mathbb{N} \}.
\end{equation}
This is a pre-Hilbert space when endowed with the inner product
\begin{equation}\label{Eq:ProdH0}
\langle f, g \rangle_K =\sum_{i,j} \alpha_i\beta_j K(t_i,s_j),
\end{equation}
where $f(\cdot) = \sum_i \alpha_i K(t_i,\cdot)$ and $g(\cdot) = \sum_j \beta_j K(s_j,\cdot)$. Note that, as $K$ is assumed to be strictly positive definite, the elements of $\mathcal{H}_0(K)$ have a unique representation in terms of $K$.

Now, the RKHS associated with $K$ is just defined as the completion $\mathcal{H}(K)$ of $\mathcal{H}_0(K)$. More precisely, the RKHS is the set of functions $f:[0,1]\to\mathbb{R}$ that are the $t$-pointwise limit of some Cauchy sequence in $\mathcal{H}_0(K)$ (see \cite{berlinet2004}, p.~18). The corresponding inner product in $\mathcal{H}(K)$ is also denoted $\langle\cdot,\cdot\rangle_K$. 

The term ``reproducing'' in the name of these spaces is after the following ``reproducing property'',
$$
f(t) = \langle f, K(t,\cdot) \rangle_K, \mbox { for all } f\in\mathcal{H}(K), \ t\in[0,1].
$$

To see the connection with the Mahalanobis distance, let us consider a random vector $(X(t_1),\ldots,X(t_d))$, instead of the whole stochastic process $X(t)$, $t\in [0,1]$. The covariance function $K(s,t)$ would be then replaced with the covariance matrix $\Sigma$  whose $(i,j)$-entry is $K(t_i,t_j)$.    From the Moore-Aronszajn Theorem we know that there exists a unique RKHS, $\mathcal{H}(\Sigma)$, in $\mathbb{R}^d$  whose reproducing kernel is $\Sigma$ see, \cite{hsing2015}, p.47--49 or \cite{berlinet2004}, p. 19.

   From the definition \eqref{Eq:DefH0} of $\mathcal{H}_0(\Sigma)$ it is clear that, in this case, this space is just  the image of the linear application defined by $\Sigma$, that is, it consists of the vectors that can be written as $x = \Sigma a$ for some $a\in\mathbb{R}^d$. Moreover, according to \eqref{Eq:ProdH0}, the inner product  between two elements $x=\Sigma a$ and $y=\Sigma b$ of this space is given by $\langle x, y \rangle_\Sigma = a'\Sigma b$. On the other hand, since $\mathcal{H}_0(\Sigma)$ is here a finite-dimensional space, it agrees with its completion $\mathcal{H}(\Sigma)$.
   
   If we assume that $\Sigma$ has full rank (if not, the generalized inverse should be used), this product can be rewritten as
 $$\langle x, y \rangle_\Sigma \ = \ a'\Sigma b \ = \ a' \Sigma \Sigma^{-1} \Sigma b \ = \ x'\Sigma^{-1}y.$$

Then, the squared distance between two vectors $x,y\in\mathcal{H}(\Sigma)$ associated  with this inner product can be expressed as 
\begin{equation}\label{eq:statement}
\Vert x - y \Vert_\Sigma^2 \ = \ \langle x-y, x-y \rangle_\Sigma \ = \ (x-y)'\Sigma^{-1}(x-y)=\sum_{i=1}^d \frac{((x-m)'e_i)^2}{\lambda_i},
\end{equation} where in the last equality we have used the second equation in \eqref{eq:inv_finito}.

We might summarize the above elementary discussion in the following statements:

\

\noindent \textit{(a) The RKHS distance $\Vert x-y\Vert_\Sigma$ in the RKHS associated with a finite-dimensional covariance operator, given by a positive definite matrix $\Sigma$, can be expressed as a simple sum involving the inverse eigenvalues of $\Sigma$, as shown in \eqref{eq:statement}.} 
	
\noindent \textit{(b) Such RKHS distance coincides with the corresponding Mahalanobis distance between $x$ and $y$.}

\

At this point it is interesting to note that the above statement (a) can be extended to the infinite-dimensional case, as pointed out in the following lemma.

\begin{lemma}\label{Lemma:NormRKHS}
	Let $\lambda_1\geq\lambda_2\geq\ldots$ be the positive eigenvalues of the integral operator associated with the kernel $K$. Let us denote by $e_i$ the corresponding unit eigenfunctions.
	For $x\in\mathcal{H}(K)$,
	\begin{equation}\label{Eq:NormRKHS}
	\Vert x \Vert_K^2 \ = \ \sum_{i=1}^\infty \frac{\langle x, e_i \rangle^2}{\lambda_i},
	\end{equation}
	and then the RKHS can be also rewritten as
	$$\mathcal{H}(K) = \{ \ x \in L^2[0,1] \ : \ \sum_{i=1}^\infty \frac{\langle x, e_i \rangle^2}{\lambda_i} < \infty \ \}.$$
	In particular, the functions $\{\sqrt{\lambda_i}e_i\}$ are an orthonormal basis for $\mathcal{H}(K)$.
\end{lemma}
\begin{proof}
	This result is just a rewording of the following theorem, whose proof can be found in \cite{amini2012}:

\
	
\noindent	\it Theorem.- 
		 Under the indicated conditions, the RKHS associated with $K$ can be written 
		\begin{equation}\label{Eq:Amini}
		\mathcal{H}(K) = \{ x \in L^2[0,1] \ : \ x  =\sum_{i=1}^\infty a_i \sqrt{\lambda_i}e_i, \ \text{for} \ \sum_{i=1}^\infty a_i^2 < \infty\},
		\end{equation}
		where the convergence of the series is in $L^2[0,1]$. This space is endowed with the inner product $\langle x, y \rangle_K \ = \ \sum_i a_i b_i,$
		where $x = \sum_i a_i \sqrt{\lambda_i}e_i$ and $y= \sum_i b_i \sqrt{\lambda_i}e_i$.
	\rm

\	
	
The result follows by noting that
	for any $x\in L^2[0,1]$ we can write
	$$x \ = \ \sum_{i=1}^\infty \langle x, e_i\rangle e_i \ = \ \sum_{i=1}^\infty \frac{\langle x, e_i\rangle}{\sqrt{\lambda_i}} \sqrt{\lambda_i}e_i.$$

Then, if the coefficients $\langle x, e_i\rangle$ tend to zero fast enough so that $\sum_i \langle x, e_i\rangle^2 \lambda_i^{-1} < \infty$, we have $x\in {\mathcal H}(K)$ and we get the expression \eqref{Eq:NormRKHS} for $\Vert x\Vert_K^2$.  	
\end{proof}

This result sheds some light on the following crucial question: to what extent the formal expression \eqref{eq:template} can be used to give a general definition of the functional Mahalanobis distance? In other words, for which functions $x\in L^2[0,1]$ does the series in \eqref{eq:template} converge in $L^2$?
The answer is  clear in view of Lemma \ref{Lemma:NormRKHS}: \textbf{expression  \eqref{eq:template} is well defined if and only if $\boldsymbol{x\in{\mathcal H}(K)}$}. This amounts to ask for a strong, very specific, regularity condition on $x$. 

The bad news is that, as a consequence of a well-known result (see, e.g. \cite{lukic2001}), Cor. 7.1) if $X=X(t)$ is a Gaussian process with mean  and covariance functions $m$ and $K$, respectively, such that $m\in{\mathcal H}(K)$ and 
${\mathcal H}(K)$ is infinite-dimensional, then ${\mathbb P}(X(\cdot)\in {\mathcal H}(K))=0$, whenever the probability ${\mathbb P}$ is assumed to be complete. 

Hence, with probability one, expression \eqref{eq:template} \textbf{is not convergent for the trajectories drawn from the stochastic process $X$}.


\subsection{The proposed definition}

In view of the discussion above (see statement (b) before Lemma \ref{Lemma:NormRKHS}), it might seem natural to define the (square) Mahalanobis functional distance between a trajectory $x=x(t)$ of the process $X(t)$ and a function $m\in L^2[0,1]$ by $M^2(x,m)=\Vert x-m\Vert_K^2$. However, this idea does not work since, as indicated above, the trajectories $x=x(t)$ of $X=X(t)$ do not belong to ${\mathcal H}(K)$ with probability one.  

 This observation suggest us the simple strategy we will follow here: given two functions $x,m\in L^2[0,1]$, just approximate them by two other functions $x_\alpha, m_\alpha\in {\mathcal H}(K)$ and calculate the distance
$\Vert x_\alpha-m_\alpha\Vert_K$. It only remains to decide how to obtain the RKHS approximations $x_\alpha$ and $m_\alpha$. One could think of taking $x_\alpha$ as the ``closest'' function to $x$ in ${\mathcal H}(K)$ but this approach also fails since ${\mathcal H}(K)$ is dense in $L^2[0,1]$  whenever all $\lambda_i$ are strictly greater than zero (Remark 4.9 of \cite{cucker2007}). Thus,  every function $x\in L^2[0,1]$ can be arbitrarily well approximated by functions in ${\mathcal H}(K)$.

This leads us in a natural way to the following penalization approach.  Let us fix a penalization parameter $\alpha>0$. Given any $x\in L^2[0,1]$,  define 
\begin{equation}\label{eq:xalpha}
x_\alpha=\argmin_{f\in {\mathcal H}(K) }\Vert x-f\Vert_2^2+\alpha\Vert f\Vert_K^2.
\end{equation}
As we will see below, the ``penalized projection'' $x_\alpha$ is well-defined. In fact it admits a relatively simple closed form. Finally, the definition we propose for the functional $\alpha$-Mahalanobis distance is
\begin{equation}\label{eq:alpha_Mah}
M_\alpha(x,m)=\Vert x_\alpha-m_\alpha\Vert_K.
\end{equation}

As mentioned, given a realization $x$ of the stochastic process we have relatively simple expressions for both the smoothed trajectory $x_\alpha$ and the proposed distance. In the next result we summarize these expressions. 

\begin{proposition}\label{Prop:PropertiesProposal}
Given a second order process with covariance $K$, we denote as $\operK$ the integral covariance operator of Equation \eqref{Eq:OpCov} associated with $K$. Then the smoothed trajectories $x_\alpha$ defined in (\ref{eq:xalpha}) satisfy the following basic properties:

\begin{enumerate}
	\item[(a)] Let $\I$ be the identity operator on $L^2[0,1]$. Then, $\operK + \alpha \I$ is invertible and 
	\begin{equation}
	\label{Eq:MinEspec}
	x_\alpha = \left(\operK + \alpha \I\right)^{-1} \operK x  = \sum_{j=1}^\infty \frac{\lambda_j}{\lambda_j+\alpha} \langle x,e_j\rangle_{2}\, e_j,
\end{equation}
where $\lambda_j$, $j=1,2,\cdots$ are the eigenvalues of $\operK$ (which are strictly positive under our assumptions) and  $e_j$ stands for the unit eigenfunction of $\operK$ corresponding to $\lambda_j$.
	  
	\item[(b)] Denoting as $\operK^{1/2}$ the square root operator defined by $(\operK^{1/2})^2 = \operK$, the norm of $x_\alpha$ in $\calh(K)$ satisfies
	\begin{equation}\label{Eq:MahSpectra}
	\|x_\alpha\|^2_K = \sum_{j=1}^\infty \frac{\lambda_j}{(\lambda_j + \alpha)^2} \ \langle x,e_j \rangle_{2}^2=\|\operK^{1/2} (\operK + \alpha\I)^{-1} x\|_2^2,
	\end{equation} 
	and therefore, 
	$$M_\alpha(x,m)^2 = \sum_{j=1}^\infty \frac{\lambda_j}{(\lambda_j + \alpha)^2} \ \langle x-m,e_j \rangle_{2}^2.$$
\end{enumerate}
\end{proposition}

\begin{proof} (a)  The fact that $\operK + \alpha \I$ is invertible is a consequence of Theorem 8.1 in \cite[p. 183]{Gohberg2013}. The expression for $x_\alpha$ follows straightforwardly from Proposition 8.6 of \cite[p.139]{cucker2007}.
Moreover, expression (8.4) in \cite[p. 184]{Gohberg2013} yields 
\begin{equation}
\label{eq:aux_lemma1}
 \left(\operK + \alpha \I\right)^{-1} y = \frac{1}{\alpha}(\I - \operK_1)y,
\end{equation}
where
\begin{equation}
\label{eq:aux_lemma2}
\operK_1 y = \sum_{j=1}^\infty \frac{\lambda_j}{\alpha + \lambda_j} \langle y, e_j\rangle_2\, e_j.
\end{equation}
Then, using the Spectral theorem for compact and self-adjoint operators (for instance Theorem 2 of Chapter 2 of \cite{cucker2001}) we get:
\[
x_\alpha =\left(\operK + \alpha \I\right)^{-1} \operK x = \frac{1}{\alpha} \sum_{j=1}^\infty \left(1-\frac{\lambda_j}{\alpha + \lambda_j}\right)\lambda_j \langle x, e_j\rangle_2\, e_j = \sum_{j=1}^\infty \frac{\lambda_j}{\alpha + \lambda_j} \langle x, e_j\rangle_2\, e_j.
\]

\

\noindent (b) In Lemma \ref{Lemma:NormRKHS} we have seen that $\sqrt{\lambda_j}e_j$  is an orthonormal basis of $\calh(K)$. Then \eqref{Eq:MinEspec} together with Parseval's identity (in $\calh(K)$) imply 
\[
\|x_\alpha\|^2_K = \sum_{j=1}^\infty \frac{\lambda_j}{(\lambda_j + \alpha)^2} \ \langle x,e_j \rangle_{2}^2.
\]
Moreover,  from the Spectral Theorem $\operK^{1/2}(x)=\sum_{i=1}^\infty \sqrt{\lambda_i} \langle x,e_i\rangle_2\,e_i$, then  using (\ref{eq:aux_lemma1}) and (\ref{eq:aux_lemma2}),
$
\operK^{1/2}(\operK + \alpha \I)^{-1} = \alpha^{-1}\operK^{1/2}(\I - \operK_1)$, and also
\[
\operK^{1/2}(\operK + \alpha \I)^{-1}x = \sum_{j=1}^\infty \frac{\sqrt{\lambda_j}}{\lambda_j+\alpha} \langle x,e_j\rangle_{2}\, e_j.
\]
Then, using again Parseval's identity (but now in $L^2[0,1]$) we get
\[
\|\operK^{1/2}(\operK + \alpha \I)^{-1}x\|^2_2 = \sum_{j=1}^\infty \frac{\lambda_j}{(\lambda_j + \alpha)^2} \ \langle x,e_j \rangle_{2}^2=\|x_\alpha\|^2_K.
\]
\end{proof}

\begin{corollary}\label{cor:metric}
	The expression $M_\alpha$ given in \eqref{eq:alpha_Mah} defines a metric in $L^2[0,1]$. 
	\begin{proof}
		This result is a direct consequence of Proposition \ref{Prop:PropertiesProposal}. Indeed, from expression \eqref{Eq:MinEspec}, the transformation
		$x\mapsto x_\alpha$ form $L^2[0,1]$ to ${\mathcal H}(K)$ in injective (since the coefficients $\langle x,e_i\rangle_K$ completely determine $x$). This, together with the fact that $\Vert\cdot\Vert_K$ is a norm, yields the result. 
	\end{proof}
\end{corollary}


\begin{remark}\label{meaning-x_alpha}
	The expression $x_\alpha= \left(\operK + \alpha \I\right)^{-1} \operK x$ obtained in the first part of Proposition \ref{Prop:PropertiesProposal} has an interesting intuitive meaning: the transformation $x\mapsto \operK x$ takes first the function $x\in L^2[0,1]$ to the space ${\mathcal H}(K)$, made of much nicer functions, with Fourier coefficients $\langle x,e_i\rangle_2$ converging quickly to zero, since we must have  $\sum_{i=1}^\infty \langle x,e_i\rangle_2^2/\lambda_i<\infty$;  see  \eqref{Eq:Amini}. Then, after this ``smoothing step'', we perform an ``approximation step'' by applying the inverse operator $\left(\operK + \alpha \I\right)^{-1}$, in order the get, as a final output, a function $x_\alpha$ that is both, close to $x$ and smoother than $x$.
	Note also that the operator $\left(\operK + \alpha \I\right)^{-1} \operK$ is compact. Thus, if we assume that the original trajectories are uniformly bounded in $L^2[0,1]$, the final result of applying on these trajectories the transformation $x\mapsto x_\alpha$ would be to take them to a  pre-compact set  of $L^2[0,1]$. This is very convenient from different points of view (beyond our specific needs here), in particular when one needs to find a convergent subsequence inside a given bounded sequence of $x_\alpha$'s.
\end{remark}
 \subsection{Some previous proposals}

Motivated by the heuristic spectral version \eqref{eq:template} of the Mahalanobis distance,  \cite{galeano2015} have proposed the following definition, that avoids the convergence problems of the series in \eqref{eq:template} (provided that $\lambda_i>0$)
at the expense of introducing a sort of smoothing parameter $k\in{\mathbb N}$,
\begin{equation}\label{eq:def_galeano}
d_{FM}^k(x,m)\ = \ \left(\sum_{i=1}^k \frac{\langle x-m,e_i\rangle^2}{\lambda_i}\right)^{1/2}.
\end{equation}
We keep the notation $d_{FM}^k$ used in \cite{galeano2015}. Let us note that $d_{FM}^k(x,m)$ is a semi-distance, since it lacks the identifiability condition $d_{FH}^k(x,m)=0 \Rightarrow x=m$. The applications of $d_{FM}^k$ considered by these authors focus mainly on supervised classification. While this proposal is quite simple and natural, it suffers from some insufficiencies when considered from the theoretical point of view. The most important one is the fact that the series \eqref{eq:def_galeano} is divergent, with probability one, whenever $x$ is a trajectory of a Gaussian process with mean function $m$ and covariance function $K$ (as we have just seen). So, $d^k_{FM}$ is defined in terms of the $k$-th partial sum of a divergent series. As a consequence, one may expect that the definition might be strongly influenced by the choice of $k$. As we will discuss below, in practice this effect is not noticed if $x$ is 
replaced with a smoothed trajectory but, in that case, the smoothing procedure should be incorporated to the definition.

Another recent proposal is due to \cite{ghi17}. The idea is also to modify the template \eqref{eq:template} to deal with the convergence issues. In this case, the suggested definition is
\begin{equation}\label{eq:def_ghi}
d_p(x,m)\ = \ \left(\int_0^\infty \sum_{i=1}^\infty \frac{\langle x-m,e_i\rangle^2}{e^{\lambda_ic}}g(c;p)dc\right)^{1/2},
\end{equation}
where $p>0$ and $g(c;p)$ is a weight function such that $g(0;p)=1$, $g$ is non-increasing and non-negative and $\int_0^\infty g(c;p)dc=p$. Moreover, for any $c>0$, $g(c;p)$ is assumed to be non-decreasing in $p$ with $\lim_{p\to\infty}g(c;p)=1$. This definition does not suffer from any problem derived from degeneracy but, still, it depends from two smoothing functions: the exponential in the denominator of \eqref{eq:def_ghi} and the weighting function $g(c;p)$. As pointed out also in \cite{ghi17}, a more convenient expression for \eqref{eq:def_ghi} is given by the following weighted version of the template, formal definition \eqref{eq:template},
\begin{equation}\label{eq:ghi_weighted}
d_p(x,m)\ = \ \left( \sum_{i=1}^\infty \frac{\langle x-m,e_i\rangle^2}{\lambda_i}h_i(p)\right)^{1/2},
\end{equation}
where $h_i(p)=\int_0^\infty \lambda_ie^{-\lambda_ic}g(c;p)dc$.

The applications of \eqref{eq:def_ghi} offered in \cite{ghi17} and \cite{ghi17b} deal with hypotheses testing for two-sample problems of type $H_0:\,m_1=m_2$.


\section{Some properties of the functional Mahalanobis distance}\label{sec:prop}

In this section we analyze in detail and prove some of the features of $M_\alpha$ we have anticipated above.  In what follows $X = X(t)$, with $t\in[0, 1]$ will stand for a second-order stochastic process with continuous trajectories and continuous mean and covariance functions, denoted by $m = m(t)$ and $K = K(s,t)$, respectively.

\subsection{Invariance}\label{subsec:inv}

In the finite dimensional case, one appealing property of the Mahalanobis distance  is the fact that it does not change if we apply a non-singular linear transformation to the data. Then, the invariance for a large class of linear operators appears also as a desirable property  for any extension of the Mahalanobis distance to the functional case. Here, we will prove invariance with respect to operators preserving the norm. We recall that an operator $L$  is an isometry if it maps $L^2[0,1]$ to $L^2[0,1]$ and $\|f\|_2 = \|Lf\|_2$. In this case, it holds $L^*L=\I$, where $L^*$ stands for the adjoint of $L$. 

\begin{theorem}\label{teo:invariance}
Let $L$ be an isometry on $L^2[0,1]$. Then, $M_\alpha(x,m)=M_\alpha(Lx,Lm)$ for all $\alpha>0$, where $M_\alpha $ was defined in \eqref{eq:alpha_Mah}.
\end{theorem}

\begin{proof}

Let $\operK_L$ be the covariance operator of the process $LX$. 
The first step of the proof is to show that $\operK_L = L\operK L^*$. It is enough to prove that for all $f,g\in L^2[0,1]$, it holds $\langle \operK_L f, g\rangle_2 = \langle L\operK L^*f, g\rangle_2$. Observe that 
\[
\langle \operK_L f, g\rangle_2 =\int_0^1 \operK_L f(t) g(t)\mathrm{d}t = \int_0^1\int_0^1 \mathbb{E}[(LX(s)-Lm(s))(LX(t)-Lm(t))]f(s)g(t)\mathrm{d}s\mathrm{d}t.
\]
Then, using Fubini's theorem and the definition of the adjoint operator:
\[
\langle \operK_L f, g\rangle_2 = \mathbb{E}\big[\langle L(X-m),  f \rangle_2 \cdot \langle L(X-m),  g \rangle_2\big]=
\mathbb{E}\big[\langle X-m,  L^*f \rangle_2 \cdot \langle X-m,  L^*g \rangle_2\big].
\]
Analogously, we also have
\[
\langle L \operK L^*f, g\rangle_2 = \langle \operK L^*f, L^*g\rangle_2 = \mathbb{E}\big[\langle X-m,  L^*f \rangle_2 \cdot \langle X-m,  L^*g \rangle_2\big].
\]
From the last two equations we conclude $\operK_L = L\operK L^*$.

The second step of the proof is to observe that the eigenvalues $\lambda_j$ of $\operK_L$ are the same as those of $\operK$, and  the unit eigenfunction $u_j$ of $\operK_L$ for the eigenvalue $\lambda_j$ is given by $v_j = L e_j$, where $e_j$ is the unit eigenfunction corresponding to $\lambda_j$. Indeed, using $L^*L=\I$ we have
\[
\operK_L v_j =  L\operK L^* v_j =  L\operK L^*L e_j = \lambda_j Le_j = \lambda_j v_j,\ \ j=1,2,\ldots
\]
Then, by \eqref{Eq:MahSpectra} and using that $L$ is an isometry,
\begin{align*}
M_\alpha (Lx, Lm) =& \|(Lx-Lm)_\alpha\|_{\operK_L} = \sum_{j=1}^\infty \frac{\lambda_j}{(\lambda_j + \alpha)^2} \ \langle Lx-Lm,L e_j \rangle_{2}^2\\
 =& \sum_{j=1}^\infty \frac{\lambda_j}{(\lambda_j + \alpha)^2} \ \langle x-m,e_j \rangle_{2}^2 = M_\alpha (x, m).
\end{align*}
\end{proof}

The family of isometries on $L^2[0,1]$ contains some interesting examples. For instance, all the symmetries and translations are isometries, as well as the changes between orthonormal  bases. Thus, this distance does not depend on the basis on which the data are represented.


\subsection{Distribution  for Gaussian processes}\label{subsec:dist}

We have mentioned in the introduction that the squared Mahalanobis distance to the mean for Gaussian data has a $\chi^2$ distribution with $d$ degrees of freedom, where $d$ is the dimension of the data. In the functional case, the distribution of $M_\alpha(X,m)^2$ for a Gaussian process $X$ equals that of an infinite linear combination of independent $\chi_1^2$ random variables. We prove this fact in the following result and its corollary, and also give explicit expressions for the expectation and the variance of $M_\alpha(X,m)^2$.

\begin{proposition}
\label{prop:dist}
Let $\{X_t:\, t\in [0,1]\}$ be an $L^2$ Gaussian process with mean $m$ and continuous  positive definite covariance function  $K$. Let $\lambda_1,\lambda_2,\cdots$ be the eigenvalues of $\operK$ and let $e_1,e_2,\ldots$ be the corresponding unit eigenfunctions.
\begin{enumerate}
\item[(a)] The squared Mahalanobis distance to the origin satisfies
\begin{equation}
\label{Eq:DistribMah}
M_\alpha(X,0)^2=\|X_\alpha\|^2_K = \sum_{j=1}^\infty \beta_jY_j,
\end{equation}
where $\beta_j=\lambda_j^2(\lambda_j+\alpha)^{-2}$ and  $Y_j$, $j=1,2,\cdots$,  are  non-central $\chi^2_1(\gamma_j)$ random variables with non-centrality parameter $\gamma_j=\mu_j^2/\lambda_j$, with $\mu_j := \langle m,e_j\rangle_2$.
\item[(b)] We have 
\[
\mathbb{E}\left(M_\alpha(X,0)^2\right) = \sum_{j=1}^\infty \frac{\lambda_j^2}{(\lambda_j + \alpha)^2}\left(1+\frac{\mu_j^2}{\lambda_j}\right),
\]
and 
\[
\var\left(M_\alpha(X,0)^2\right) = 2\sum_{j=1}^\infty \frac{\lambda_j^4}{(\lambda_j + \alpha)^4}\left(1+\frac{2\mu_j^2}{\lambda_j}\right).
\]
\end{enumerate}
\end{proposition} 

\begin{proof}
(a) Using \eqref{Eq:MahSpectra}, $\|X_\alpha\|^2_K = \sum_{j=1}^\infty \beta_jY_j$, where $\beta_j=\lambda_j^2(\lambda_j+\alpha)^{-2}$ and  $Y_j=\lambda_j^{-1}\langle X,e_j\rangle_2^2$. Since the process is Gaussian the variables $\lambda_j^{-1/2}\langle X,e_j\rangle$  are independent with normal distribution, mean $\lambda_j^{-1/2}\mu_j$ and variance 1 (see \cite{ash2014topics}, p. 40). The result follows.

\

\noindent (b)  It is easy to see that the partial sums in \eqref{Eq:DistribMah} form a sub-martingale with respect to the natural filtration $\sigma(Y_1, \ldots, Y_N)$,
\[
\E \left( \sum_{j=1}^{N+1} \beta_j Y_j \, \Big\vert \, Y_1, \ldots, Y_N \right) =
\beta_{N+1}\mathbb{E}(Y_{N+1}) + \sum_{j=1}^N \beta_j Y_j  \geq \ \sum_{j=1}^N \beta_j Y_j.
\]
Moreover, if $\bar\lambda := \sup_j \lambda_j$, which is always finite,
\[
\sup_N \E \Big( \sum_{j=1}^{N+1} \beta_j Y_j \Big) = \sum_{j=1}^\infty \frac{\lambda_j(\lambda_j+\mu_j^2)}{(\lambda_j+\alpha)^2}\leq \frac{\bar{\lambda}}{\alpha^2}
 \Big(\sum_{j=1}^\infty \lambda_j + \sum_{j=1}^\infty \mu_j^2 \Big)<\infty,
\]
because $m\in L^2[0,1]$ and $\sum_{j=1}^\infty \lambda_j = \int_0^1 K(t,t)\mathrm{d}t<\infty$ (see e.g. \cite{cucker2001}, Corollary 3, p. 34). Now, Doob's convergence theorem implies $\sum_{j=1}^N \beta_jY_j \to \sum_{j=1}^\infty \beta_jY_j$ a.s. as $N\to\infty$, and Monotone Convergence theorem yields the expression for the expectation of $M_\alpha(X,0)^2$.

The proof for the variance is fairly similar. Using Jensen's inequality, we deduce
\[
\E \left[ \left(\sum_{j=1}^{N+1} \beta_j \left( Y_j - \E Y_j \right) \right)^2\, \Big\vert \, Y_1, \ldots, Y_N \right] \geq \left( \sum_{j=1}^N \beta_j \left( Y_j - \E Y_j \right) \right)^2.
\]
Moreover, since the variables $Y_j$ are independent:
\begin{align*}
\sup_N \E \left( \sum_{j=1}^N \beta_j \left( Y_j - \E Y_j \right) \right)^2 =&
\sum_{j=1}^\infty \beta_j^2 \mbox{Var}(Y_j) = 2\sum_{j=1}^\infty \frac{\lambda_j^3(\lambda_j + 2\mu_j^2)}{(\lambda_j + \alpha)^4} \\
\leq& \frac{2\bar{\lambda}^3}{\alpha^4}\Big(\sum_{j=1}^\infty \lambda_j + 2\sum_{j=1}^\infty \mu_j^2 \Big)<\infty.
\end{align*}
Then, $\left(\sum_{j=1}^{N+1} \beta_j \left( Y_j - \E Y_j \right) \right)^2\to \left(\sum_{j=1}^{\infty} \beta_j \left( Y_j - \E Y_j \right) \right)^2$ a.s., as $N\to\infty$, and using Monotone Convergence theorem,
\[
\var\left(M_\alpha(X,0)^2\right) = \lim_{N\to\infty}\mbox{Var}\Big(\sum_{j=1}^N \beta_jY_j\Big) = 2\sum_{j=1}^\infty \frac{\lambda_j^4}{(\lambda_j + \alpha)^4}\left(1+\frac{2\mu_j^2}{\lambda_j}\right).
\]
\end{proof}

When we compute the squared Mahalanobis distance to the mean the expressions above simplify because $\mu_j=0$ for each $j$, and then we have the following corollary.

\begin{corollary}
\label{cor:distribution}
Under the same assumptions of Proposition \ref{prop:dist}, $M_\alpha(X,m)^2= \sum_{j=1}^\infty \beta_jY_j$, where  $\beta_j=\lambda_j^2(\lambda_j+\alpha)^2$ and $Y_1, Y_2,\ldots$ are independent  $\chi^2_1$ random variables. Moreover, $\mathbb{E}\left(M_\alpha(X,m)^2\right) = \sum_{j=1}^\infty \lambda_j^2 (\lambda_j+\alpha)^{-2}$ and $\var\left(M_\alpha(X,m)^2\right) = 2\sum_{j=1}^\infty \lambda_j^4 (\lambda_j+\alpha)^{-4}$.
\end{corollary}

\subsection{Stability  with respect to  $\alpha$}
\label{subsec:alpha}

Our definition of distance depends on a regularization parameter $\alpha>0$. In this subsection  we prove the continuity of $M_\alpha$ with respect to the tuning parameter $\alpha$. The proof of the main result requires the following auxiliary lemma, which has been adapted from Corollary 8.3 in \cite{Gohberg2013}, p. 71. Recall that given a bounded operator $A:\mathcal{H}\to\mathcal{H}$ on a Hilbert space $\mathcal{H}$ we can define the norm
\[
\|A\|_{\mathcal{L}} := \sup\{ \|Ax\|_{\mathcal{H}}:\, \|x\|_{\mathcal{H}}\leq 1\}. 
\]

\begin{lemma}
\label{lemma:aux}
Let $A_j:\mathcal{H}\to\mathcal{H}$, $j=1,2,\ldots$, be a sequence of bounded invertible operators on a Hilbert space $\mathcal{H}$ which converges in norm $\|\cdot\|_{\mathcal{L}}$ to another operator $A$, and such that $\sup_j \|A_j^{-1}\|_{\mathcal{L}}<\infty$. Then $A$ is also invertible, and  $\|A_j^{-1} - A^{-1}\|_{\mathcal{L}}\to 0$, as $j\to\infty$.  
\end{lemma}

We will apply the preceding lemma in the proof of the following result.

\begin{proposition}
\label{Prop:Cont}
Let $\alpha_j$ be a sequence of positive real numbers such that $\alpha_j\to\alpha>0$, as $j\to\infty$. Then, $\|X_{\alpha_j}\|_K \to \|X_{\alpha}\|_K$ a.s. as $j\to\infty$. 
\end{proposition}

\begin{proof}
Note that by Proposition \ref{Prop:PropertiesProposal}(b), Eq \eqref{Eq:MahSpectra}, we have
\begin{align*}
\Big\vert \|X_{\alpha_j}\|_K - \|X_{\alpha}\|_K \Big\vert & \leq \|\operK^{1/2}(\operK+\alpha_j\I)^{-1}X-\operK^{1/2}(\operK+\alpha\I)^{-1}X\|_2\\
& \leq  \|\operK^{1/2} \|_{\mathcal{L}} \ \|(\operK+\alpha_j\I)^{-1}-(\operK+\alpha\I)^{-1}\|_{\mathcal{L}} \ \|X\|_2. 
\end{align*}
But it holds
\[
\|(\operK+\alpha_j\I)-(\operK+\alpha\I)\|_{\mathcal{L}} = |\alpha_j-\alpha|\to 0,\ \  \mbox{as}\  j\to\infty,
\]
and $\sup_j \|(\operK+\alpha_j\I)^{-1}\|_{\mathcal{L}} \leq \inf_j \alpha_j  <\infty$ (see \cite{Gohberg2013}, (1.14), p. 228). Therefore, $\|(\operK+\alpha_j\I)^{-1}-(\operK+\alpha\I)^{-1}\|_{\mathcal{L}}\to 0$, as $j\to\infty$, by Lemma \ref{lemma:aux}.
\end{proof}

Observe that Proposition \ref{Prop:Cont} implies the point convergence of the sequence of distribution functions of $M_{\alpha_j}(X,m)$  to that of $M_\alpha(X,m)$. This fact in turn implies the point convergence of the corresponding quantile functions.


\section{A consistent estimator of the functional Mahalanobis distance}
\label{sec:consist}

Given a sample $X_1(t),\ldots,X_n(t)$ of realizations of the stochastic process $X(t)$, we want to estimate the Mahalanobis distance between any trajectory of the process $X$ and the mean function $m$ in a consistent way. Let $\bar{X}(t)=n^{-1}\sum_{i=1}^n X_i(t)$ be the sample mean and let
\[
\widehat{K}(s,t) = \frac{1}{n}\sum_{i=1}^n (X_i(s)-\bar{X}(s))(X_i(t)-\bar{X}(t))
\]
be the sample covariance function. The function $\widehat{K}$ defines the sample covariance operator $\widehat{\operK} f(\cdot) = \int_0^1 \wh K(\cdot, t)f(t)\dd t$. 

Define the following estimator for $M_\alpha(X,m)$: 
\begin{equation}
\label{Eq:EstimatorMah}
\widehat{M}_{\alpha,n}(X,\bar{X}) :=\|\wh{X}_\alpha - \bar{X}_\alpha\|_{\widehat K_n},
\end{equation}
where $\wh{X}_\alpha = (\wh{\operK} + \alpha\I)^{-1}\wh{\operK}X$ and $\bar{X}_\alpha = (\wh{\operK} + \alpha\I)^{-1}\wh{\operK}\bar{X}$. 

In the following Lemma we establish the  consistency in the operator norm of $\wh \operK$ as an estimator of $\operK$,
as a preliminary step to show the consistency of $\widehat{M}_{\alpha,n}$. 

\begin{lemma}
\label{lemma:consistencyK}
Suppose that $\mathbb{E}\|X\|_2^2<\infty$. Then $\|\bar{X}-m\|_2\to 0$, $\|\wh{K}-K\|_{L^2([0,1]\times [0,1])}\to 0$ and $\|\wh\operK - \operK\|_{\mathcal{L}}\to 0$, a.s. as $n\to\infty$.
\end{lemma}

\begin{proof}
Mourier's  SLLN (see e.g. Theorem 4.5.2 in \cite{laha1979}, p. 452) implies directly $\|\bar{X}-m\|_2\to 0$ since $(\E \Vert X\Vert_2)^2 \leq  \mathbb{E}\|X\|_2^2 <\infty$ and $L^2[0,1]$ is a separable Banach space.

Consider the process $Z(s,t)=X(s)X(t)$. Then,  $Z\in L^2([0,1]\times[0,1])$ and this  is also a separable Banach space. Therefore, if $Z_i(s,t)=X_i(s)X_i(t)$, $\bar{Z}=n^{-1} Z_i(s,t)$, and  $m_z(s,t)=\E [ X(s)X(t) ]$,  using again Mourier's  SLLN we have
\[
\|\bar{Z}-m_z\|_{L^2([0,1]\times [0,1])}\to 0,\ \ \mbox{a.s.}, \ \ n\to\infty,
\]
 and also, since $\wh K(s,t) = \bar{Z}(s,t)-\bar{X}(s)\bar{X}(t)$, $\|\wh\operK-\operK\|_{HS} \to 0$ a.s., where  $\|\wh\operK-\operK\|_{HS} = \|\widehat{K}-K\|_{L^2([0,1]\times [0,1])}$  stands for the Hilbert-Schmidt norm of the operator  $\wh\operK-\operK$. 

Finally, for any $x\in L^2[0,1]$, 
$$\Vert (\wh\operK - \operK) x \Vert_2^2 \ = \int_0^1 \langle \wh K(t,\cdot)- K(t,\cdot), x \rangle_{2}^2 \dd t \ \leq \ \|x\|_2 ^2 \ \Vert \wh K - K \Vert^2_{L^2([0,1]\times [0,1])}.$$
Thus, in particular, the operator norm is smaller than the Hilbert-Schmidt norm and we have  $\|\wh\operK - \operK\|_{\mathcal{L}}\leq \Vert \wh K - K \Vert_{L^2([0,1]\times [0,1])} \to 0$ a.s. as $n\to\infty$. 
\end{proof}
  
As already mentioned, by the square root of an operator $F$ we mean the operator $G$ such that $G^2 = F$.

\begin{theorem}
\label{Teo:ConsistMah}
If $\E \Vert X \Vert^2_2 <\infty$,  then  $\widehat{M}_{\alpha,n}(X,\bar{X})\to M_\alpha(X,m)$ a.s., as $n\to\infty$.
\end{theorem}

\begin{proof}
From Proposition \ref{Prop:PropertiesProposal}(b), Eq \eqref{Eq:MahSpectra}, we have
$\widehat{M}_{\alpha,n}(X,\bar{X}) = \|\wh\operK^{1/2}(\wh\operK + \alpha\I)^{-1}(X-\bar{X})\|_2$. Therefore,
\begin{align*}
& \Big\vert \widehat{M}_{\alpha,n}(X,\bar{X}) - M_\alpha(X,m)\Big\vert \leq 
\|\wh\operK^{1/2}(\wh\operK + \alpha\I)^{-1}(X-\bar{X}) - \operK^{1/2}(\operK + \alpha\I)^{-1}(X-m)\|_2\\
&\leq  \|\wh\operK^{1/2}(\wh\operK + \alpha\I)^{-1} \|_{\mathcal{L}}\ \|\bar{X}-m\|_2 +
\|\wh\operK^{1/2}(\wh\operK + \alpha\I)^{-1}  - \operK^{1/2}(\operK + \alpha\I)^{-1}\|_{\mathcal{L}}\ \|X-m\|_2.
\end{align*}
By Lemma \ref{lemma:consistencyK}, $\|\bar{X}-m\|_2$ goes to zero a.s. as $n\to \infty$. As a consequence, it is enough to show that  $\|\wh\operK^{1/2}(\wh\operK + \alpha\I)^{-1}  - \operK^{1/2}(\operK + \alpha\I)^{-1}\|_{\mathcal{L}}\to 0$ a.s.  
For that purpose, observe that
\begin{align*}
& \|\wh\operK^{1/2}(\wh\operK + \alpha\I)^{-1}  - \operK^{1/2}(\operK + \alpha\I)^{-1}\|_{\mathcal{L}} \\ 
& \leq 
 \|\wh\operK^{1/2}(\wh\operK + \alpha\I)^{-1}  - \wh\operK^{1/2}(\operK + \alpha\I)^{-1}\|_{\mathcal{L}} +
\|\wh\operK^{1/2}(\operK + \alpha\I)^{-1}  - \operK^{1/2}(\operK + \alpha\I)^{-1}\|_{\mathcal{L}}\\
& \leq 
 \|\wh\operK^{1/2}\|_{\mathcal{L}} \ \|(\wh\operK + \alpha\I)^{-1}  - (\operK + \alpha\I)^{-1}\|_{\mathcal{L}}   + \|\wh \operK^{1/2} -\operK^{1/2}\|_{\mathcal{L}} \  \|(\operK + \alpha\I)^{-1}\|_{\mathcal{L}}.
\end{align*}

Therefore, to end the proof we  will  show that   $\|\wh\operK^{1/2} -\operK^{1/2}\|_{\mathcal{L}}\to 0$ a.s. as $n\to \infty$ and $\|(\wh\operK + \alpha\I)^{-1}  - (\operK + \alpha\I)^{-1}\|_{\mathcal{L}}\to 0$ a.s. as $n\to \infty$. Since the square root is a continuous function in $[0,\infty)$, the  first  result follows from part one of Theorem VIII.20 of \cite{reed1980}. The requirement of the function vanishing at infinity is irrelevant here since, from Lemma \ref{lemma:consistencyK} we know $\|\wh\operK -\operK\|_{\mathcal{L}}\to 0$ a.s. as $n\to \infty$, which in particular implies that there exist a bound on the norm of operators $\widehat\operK$. 
 
Finally, observe that 
 $\|\wh\operK -\operK\|_{\mathcal{L}}=\|(\wh\operK + \alpha\I)  - (\operK + \alpha\I)\|_{\mathcal{L}}\to 0$ a.s., and we also have $\sup_n \| (\wh\operK + \alpha\I)^{-1}\|_{\mathcal{L}}\leq \alpha^{-1} <\infty$. Then, Lemma \ref{lemma:aux} implies  $\|(\wh\operK + \alpha\I)^{-1}  - (\operK + \alpha\I)^{-1}\|_{\mathcal{L}}\to 0$ a.s. as $n\to\infty$.
\end{proof}

\begin{corollary}\label{Cor:ConsistF}
In fact, the result is true when measuring the distance between the mean and any function $f$ in $L^2[0,1]$, that is, $\widehat{M}_{\alpha,n}(f,\bar{X})\to M_\alpha(f,m)$ a.s., as $n\to\infty$.
\end{corollary}

Putting together Theorem \ref{Teo:ConsistMah} and Proposition \ref{prop:dist} we obtain the asymptotic distribution of $\widehat{M}_{\alpha,n}$:

\begin{corollary}
\label{cor:asymptoticdist}
Under the assumptions of Theorem \ref{Teo:ConsistMah} and Corollary \ref{cor:distribution}, and with the same notation, 
$\widehat{M}_{\alpha,n}(X,\bar{X})$ converges in distribution to  $\sum_{j=1}^\infty \beta_jY_j$, where  $\beta_j=\lambda_j^2(\lambda_j+\alpha)^{-2}$ and $Y_1, Y_2,\ldots$ are independent  $\chi^2_1$ random variables.
\end{corollary}

We can also prove another consistency result involving the distances between the sample and the population means, which could be useful for doing inference on the mean.

\begin{theorem}\label{teo:consistdist}
If $\E \Vert X \Vert^2_2 <\infty$, and with the same notation of Proposition \ref{prop:dist}, it holds,
\begin{equation}\label{eq:nconsist}
\sqrt{n} \ \wh M_{\alpha,n} (\bar{X},m) \ \std \ \Big(\sum_{j=1}^\infty \frac{\lambda_j^2}{(\lambda_j+\alpha)^2} Y_j\Big)^{\frac{1}{2}},
\end{equation}
where $Y_1, Y_2,\ldots$ are independent $\chi^2_1$ random variables.
\end{theorem}

\begin{proof}
We can rewrite the left-hand side of Equation \eqref{eq:nconsist} as,
\begin{equation}\label{eq:demn}
\sqrt{n} \ \wh M_{\alpha,n} (\bar{X},m) \ = \ \sqrt{n} ( \wh M_{\alpha,n} (\bar{X},m) - M_{\alpha}(\bar{X},m)) + \sqrt{n} M_{\alpha}(\bar{X},m).
\end{equation}
Now, from Equation \eqref{eq:alpha_Mah} and Proposition \ref{Prop:PropertiesProposal}, we have

\small
\begin{eqnarray*}
&&\sqrt{n} \vert \wh M_{\alpha,n} (\bar{X},m) - M_{\alpha}(\bar{X},m) \vert \leq  \sqrt{n} \ \Vert \wh\operK^{1/2}(\wh\operK + \alpha\I)^{-1}(\bar{X}-m) - \operK^{1/2}(\operK + \alpha\I)^{-1}(\bar{X}-m) \Vert_2 \\
&& \leq \Vert \wh\operK^{1/2}(\wh\operK + \alpha\I)^{-1} - \operK^{1/2}(\operK + \alpha\I)^{-1}\Vert_{\mathcal{L}} \ \Vert \sqrt{n}(\bar{X}-m) \Vert_2.
\end{eqnarray*}
\normalsize

 As a part  of the proof of Theorem \ref{Teo:ConsistMah} we have seen that the first norm in the right-hand side goes to zero a.s. as $n\to\infty$. From the Functional Central Limit Theorem (e.g., Theorem 8.1.1 of \cite{hsing2015theoretical}), $\sqrt{n}(\bar{X}-m)$ converges in distribution in $L^2[0,1]$ to a Gaussian stochastic process $Z$ with zero mean and covariance operator $\operK$. Since the norm is a continuous function in this space, by the continuous mapping theorem the second term converges in distribution to the random variable $\Vert Z \Vert_2$. Thus, by  Slutsky's  theorem, the distribution of the product goes to zero, and this convergence holds also in probability since the limit is a constant.

We can rewrite the remaining term of Equation \eqref{eq:demn} as,
$$\sqrt{n} M_{\alpha}(\bar{X},m) \ = \ \sqrt{n} \ \Vert \operK^{1/2}(\operK + \alpha\I)^{-1}(\bar{X}-m) \Vert_2  \ = \ \Big\Vert \sqrt{n} \ \Big( \frac{1}{n}\sum_{i=1}^n \chi_{\alpha,i} - \mu_\alpha \Big) \Big\Vert_2,$$
where we denote $\chi_{\alpha,i} = \operK^{1/2}(\operK + \alpha\I)^{-1} X_i$ and $\mu_{\alpha} = \operK^{1/2}(\operK + \alpha\I)^{-1} m$. Since $\operK^{1/2}(\operK + \alpha\I)^{-1}$ is a bounded linear operator and the process $X$ is Bochner-integrable ($\E \Vert X \Vert_2<\infty$), the expectation and the operator commute, that is,
$$\E \chi_{\alpha,1} \ = \ \E [\operK^{1/2}(\operK + \alpha\I)^{-1} X_1] \ = \ \operK^{1/2}(\operK + \alpha\I)^{-1} \E X_1 \ = \ \mu_\alpha.$$
Therefore, we can use again the Functional Central Limit Theorem with $\chi_{\alpha,i}$ and $\mu_\alpha$, since
$$\E \Vert \chi_{\alpha,1} \Vert_2^2 \ \leq \ \Vert \operK^{1/2}(\operK + \alpha\I)^{-1} \Vert_{\mathcal{L}}^2 \ \E \Vert X_1\Vert_2^2 \ < \ \infty,$$
which gives us that $\sqrt{n} M_{\alpha}(\bar{X},m)$ converges in distribution to $\Vert \xi\Vert_2$,  $\xi$ being  a random element with zero mean and  whose covariance operator is the same as that of $\chi_{\alpha,1}$.

Using the same reasoning as at beginning of the proof of Theorem \ref{teo:invariance}  and denoting as $A^*$ the adjoint of the operator $A$, the covariance operator of $\chi_{\alpha,1}$ is given by 
$$\operK^{1/2}(\operK + \alpha\I)^{-1} \operK [\operK^{1/2}(\operK + \alpha\I)^{-1}]^* \ = \ \operK^{1/2}(\operK + \alpha\I)^{-1} \operK (\operK + \alpha\I)^{-1} \operK^{1/2},$$
since both $\operK^{1/2}$ and $(\operK + \alpha\I)^{-1}$ are self-adjoint operators  (for instance, Theorem 3.35 and Problem 3.32 of \cite[Chapter 5]{kato2013} and Proposition 2.4 of \cite[Chapter X]{conway1990}). Now since $\xi$ is a zero-mean Gaussian process with compact covariance operator, it has an associated orthonormal basis of eigenfunctions (Spectral theorem for compact and self-adjoint operators, for instance Theorem 2 of Chapter 2 of \cite{cucker2001}). This operator has the same eigenfunctions as $\operK$ and its eigenvalues are $\lambda_j^2(\lambda_j+\alpha)^{-2}$. Thus, using its  Karhunen-Lo\`eve  representation we get
$$\Vert \xi\Vert_2 \ = \ \Vert \sum_{j=1}^\infty Z_j e_j \Vert_2 \ = \ \Big(\sum_{j=1}^\infty Z_j^2\Big)^{\frac{1}{2}},$$
where $e_j$ are the eigenfunctions of $\operK$ and $Z_j$ are independent Gaussian random variables with zero mean and  variances  $\lambda_j^2(\lambda_j+\alpha)^{-2}$ (the eigenvalues of the covariance operator of $\xi$). Then the result follows from the standardization of these $Z_j$, applying Slutsky's theorem to the sum of Equation \eqref{eq:demn}.
\end{proof}


\section{Statistical applications}\label{sec:appl}

The purpose of this section is to give a general overview of possible applications of the proposed distance. The selected models have been mostly chosen among those previously proposed  in the literature. However, as usual in empirical studies, many other meaningful scenarios could be considered. Thus we make no attempt to reach any definitive conclusion. Only the long term practitioners' experience will lead to a safer judgment.


\subsection{Exploratory analysis}

The Mahalanobis distance can be used to analyze and summarize some interesting features of the data which, for instance, can be done by generating boxplots. We follow here the experimental setting proposed in \cite{arribas2014}, where some real and simulated data sets are used for outliers detection and functional boxplots.

\

\noindent\textit{Outliers detection}

The simulation study proposed in \cite{arribas2014} checks the performance of ten different methods. The curves are generated using three different combinations of  the main process (from which most trajectories are drawn) and the contamination one (from which the outliers come from). Given a contamination rate $c$, $n-\lceil c\cdot n\rceil$ curves are drawn from the main process and $\lceil c\cdot n\rceil$ from the contamination one (we denote as $\lceil x \rceil$ the smallest integer not  smaller  than $x$).
\begin{itemize}
\item  The   first model is defined by, \\[-7mm]
\begin{itemize}
\item[] $\text{main process: } X(t) \ = \ 30t(1-t)^{3/2} + \varepsilon(t),$ \\[-7mm]
\item[] $\text{contamination process: } X(t) \ = \ 30t^{3/2}(1-t) + \varepsilon(t),$\\[-7mm]
\end{itemize}
for $t\in[0,1]$, where $\varepsilon$ is a Gaussian process with zero mean and covariance function $K(s,t) = 0.3 \ \text{exp}(-|s-t|/0.3)$. 
\item  The second model is given by,\\[-7mm]
\begin{itemize}
\item[] $\text{main process: } X(t) \ = \ 4t + \varepsilon(t),$ \\[-7mm]
\item[] $\text{contamination process: } X(t) \ = \ 4t + (-1)^u1.8 + (0.02\pi)^{-1/2} \  e^{\frac{-(t-\mu)^2}{0.02}} + \varepsilon(t),$\\[-7mm]
\end{itemize}
where $\varepsilon$ is a Gaussian process with zero mean and covariance function $K(s,t) = \text{exp}(-|s-t|)$, $u$ follows a Bernoulli distribution with parameter $0.5$ and $\mu$ is uniformly distributed over $[0.25,0.75]$. 
\item Finally, using the same definitions for $\varepsilon$ and $\mu$, the third model is given by,\\[-7mm]
\begin{itemize}
\item[] $\text{main process: } X(t) \ = \ 4t + \varepsilon(t),$ \\[-7mm]
\item[] $\text{contamination process: } X(t) \ = \ 4t + 2\sin(4(t+\mu)\pi) + \varepsilon(t).$\\[-7mm]
\end{itemize}
\end{itemize}

We ran 100 simulations of each model with different contamination rates $c = 0$, $0.05$, $0.1$, $0.15$ and $0.2$. The sample size  for each simulation was  $100$ and the curves are   simulated in a discretized fashion  over a grid of $50$ equidistant points in $[0,1]$. We have checked nine out of the ten methods exposed in \cite{arribas2014}, whose code is provided by the authors. The details about the implementations of  each method  can be found on the paper. We have adapted the code provided by the authors to include our method. 

In order to formally define what  we exactly mean by ``an outlier'' in our case, we have approximated the distribution of the random variable  $\Vert X_\alpha-m_\alpha\Vert_K$ given in Corollary \ref{cor:distribution}  through a Monte Carlo sample of size 2000 where the Monte Carlo observations are generated using the covariance structure of the original data.

 Then we mark as outliers the curves whose distance to the mean is greater that the $95\%$ of the  distances for the simulated data. The main drawback of this method is that the distribution of  Corollary \ref{cor:distribution}  is computed using the covariance structure of the data. Therefore, if the number of outliers is  large  compared with the sample size, this  estimation  is biased. In order to partially overcome this problem, we compute the covariance function using the robust minimum covariance determinant (MCD) estimator.
 
 Regarding our proposal,  we have noticed that  the choice of $\alpha$ does not affect the number of selected outliers significantly. We have chosen $\alpha=0.01$, but  an automatic technique (as as the one proposed in \cite{arribas2014}  for the choice  of the factor of the adjusted outliergram) could be used as well.  

The rates of correct ($p_c$) and false ($p_f$) outliers detected for each method on the different settings can be found in Table \ref{Table:Outs}.  An extended version of this table can be found in the Supplementary Material document.  We can see that the Mahalanobis-based method proposed in this paper  (denoted \textit{Mah RKHS} in the table)  is quite competitive.

\begin{table}[ht]
\centering
\resizebox{\textwidth}{!}{
\begin{tabular}{lcccccc}
  \hline
  \hline
c= 0 & Model  1 &  & Model  2 &  & Model  3 &  \\ \hline
  Method & $p_c$ & $p_f$ & $p_c$ & $p_f$ & $p_c$ & $p_f$ \\ \hline
  Adj. Fun. BP & - & 0.007  ( 0.010) & - & 0.006  ( 0.010) & - & 0.007  ( 0.012) \\ 
  Rob. Mah. Dist. & - & 0.016  ( 0.014) & - & 0.015  ( 0.013) & - & 0.015  ( 0.015) \\ 
  ISE & - & 0.038  ( 0.020) & - & 0.032  ( 0.021) & - & 0.033  ( 0.021) \\ 
  DB trimming & - & 0.013  ( 0.007) & - & 0.012  ( 0.006) & - & 0.014  ( 0.007) \\ 
  Adj. Ourliergram & - & 0.012  ( 0.012) & - & 0.011  ( 0.013) & - & 0.011  ( 0.014) \\ 
  Mah. RKHS & - & 0.037  ( 0.015) & - & 0.033  ( 0.018) & - & 0.035  ( 0.016) \\ \hline
 c= 0.05 & Model  1 &  & Model  2 &  & Model  3 &  \\ \hline
  Method & $p_c$ & $p_f$ & $p_c$ & $p_f$ & $p_c$ & $p_f$ \\
  \hline
  Adj. Fun. BP & 0.576  ( 0.282) & 0.008  ( 0.012) & 0.551  ( 0.330) & 0.006  ( 0.010) & 0.588  ( 0.344) & 0.008  ( 0.012) \\ 
  Rob. Mah. Dist. & 0.976  ( 0.096) & 0.008  ( 0.009) & 0.361  ( 0.250) & 0.008  ( 0.010) & 0.104  ( 0.153) & 0.015  ( 0.013) \\ 
  ISE & 0.865  ( 0.313) & 0.033  ( 0.020) & 1.000  ( 0.000) & 0.038  ( 0.026) & 1.000  ( 0.000) & 0.033  ( 0.021) \\ 
  DB trimming & 0.947  ( 0.183) & 0.008  ( 0.009) & 0.957  ( 0.135) & 0.008  ( 0.009) & 0.994  ( 0.035) & 0.006  ( 0.007) \\ 
  Adj. Ourliergram & 0.994  ( 0.035) & 0.006  ( 0.008) & 0.978  ( 0.070) & 0.006  ( 0.009) & 0.998  ( 0.020) & 0.012  ( 0.014) \\  
  Mah. RKHS & 0.998  ( 0.020) & 0.022  ( 0.016) & 1.000  ( 0.000) & 0.027  ( 0.014) & 1.000  ( 0.000) & 0.031  ( 0.016) \\ 
  \hline
  c= 0.1 & Model  1 &  & Model  2 &  & Model  3 &  \\ \hline
  Method & $p_c$ & $p_f$ & $p_c$ & $p_f$ & $p_c$ & $p_f$ \\ 
  \hline
  Adj. Fun. BP & 0.549  ( 0.239) & 0.005  ( 0.008) & 0.593  ( 0.268) & 0.008  ( 0.010) & 0.632  ( 0.248) & 0.008  ( 0.012) \\ 
  Rob. Mah. Dist. & 0.961  ( 0.105) & 0.004  ( 0.007) & 0.373  ( 0.170) & 0.007  ( 0.009) & 0.104  ( 0.108) & 0.011  ( 0.014) \\ 
  ISE & 0.790  ( 0.335) & 0.027  ( 0.017) & 1.000  ( 0.000) & 0.036  ( 0.021) & 1.000  ( 0.000) & 0.033  ( 0.022) \\ 
  DB trimming & 0.808  ( 0.340) & 0.009  ( 0.009) & 0.989  ( 0.045) & 0.010  ( 0.010) & 0.995  ( 0.030) & 0.008  ( 0.011) \\ 
  Adj. Ourliergram & 0.897  ( 0.118) & 0.006  ( 0.009) & 0.971  ( 0.076) & 0.006  ( 0.009) & 1.000  ( 0.000) & 0.007  ( 0.011) \\ 
  Mah. RKHS & 0.767  ( 0.148) & 0.014  ( 0.012) & 1.000  ( 0.000) & 0.014  ( 0.011) & 0.995  ( 0.030) & 0.015  ( 0.013) \\ 
  \hline
  c= 0.15 & Model  1 &  & Model  2 &  & Model  3 &  \\ \hline
  Method & $p_c$ & $p_f$ & $p_c$ & $p_f$ & $p_c$ & $p_f$ \\ 
  \hline
  Adj. Fun. BP & 0.494  ( 0.215) & 0.006  ( 0.010) & 0.550  ( 0.242) & 0.006  ( 0.009) & 0.584  ( 0.247) & 0.006  ( 0.009) \\ 
  Rob. Mah. Dist. & 0.927  ( 0.098) & 0.001  ( 0.003) & 0.324  ( 0.184) & 0.004  ( 0.007) & 0.152  ( 0.175) & 0.005  ( 0.008) \\ 
  ISE & 0.778  ( 0.349) & 0.027  ( 0.018) & 0.999  ( 0.007) & 0.040  ( 0.029) & 1.000  ( 0.000) & 0.034  ( 0.023) \\ 
  DB trimming & 0.444  ( 0.410) & 0.009  ( 0.011) & 0.981  ( 0.099) & 0.016  ( 0.015) & 0.993  ( 0.067) & 0.009  ( 0.011) \\ 
  Adj. Ourliergram & 0.616  ( 0.220) & 0.003  ( 0.007) & 0.969  ( 0.099) & 0.006  ( 0.008) & 0.996  ( 0.019) & 0.007  ( 0.010) \\ 
  Mah. RKHS & 0.295  ( 0.122) & 0.013  ( 0.011) & 0.988  ( 0.052) & 0.008  ( 0.009) & 0.941  ( 0.167) & 0.007  ( 0.009) \\ 
  \hline
  c= 0.2 & Model  1 &  & Model  2 &  & Model  3 &  \\ \hline
  Method & $p_c$ & $p_f$ & $p_c$ & $p_f$ & $p_c$ & $p_f$ \\ 
  \hline
  Adj. Fun. BP & 0.376  ( 0.226) & 0.003  ( 0.006) & 0.509  ( 0.205) & 0.005  ( 0.009) & 0.540  ( 0.227) & 0.003  ( 0.006) \\ 
  Rob. Mah. Dist. & 0.866  ( 0.167) & 0.000  ( 0.002) & 0.304  ( 0.171) & 0.002  ( 0.005) & 0.111  ( 0.118) & 0.004  ( 0.007) \\ 
  ISE & 0.513  ( 0.396) & 0.031  ( 0.023) & 0.997  ( 0.018) & 0.047  ( 0.031) & 0.999  ( 0.010) & 0.028  ( 0.023) \\ 
  DB trimming & 0.235  ( 0.314) & 0.009  ( 0.013) & 0.990  ( 0.037) & 0.015  ( 0.014) & 0.979  ( 0.121) & 0.011  ( 0.011) \\ 
  Adj. Ourliergram & 0.248  ( 0.179) & 0.001  ( 0.003) & 0.959  ( 0.074) & 0.004  ( 0.008) & 0.999  ( 0.007) & 0.008  ( 0.011) \\ 
  Mah. RKHS & 0.141  ( 0.089) & 0.012  ( 0.011) & 0.945  ( 0.127) & 0.005  ( 0.007) & 0.749  ( 0.232) & 0.006  ( 0.009) \\ 
   \hline\hline
\end{tabular}}
\caption{Ratio of correct and false detected outliers.} 
\label{Table:Outs}
\end{table}

\

\noindent\textit{Boxplots}

As a part of the exploratory analysis of the data, we include the functional boxplots of two real data sets used also in \cite{arribas2014}.
\begin{itemize}
\item Male mortality rates in Australia 1901-2003: this data set can be found in  the  R package ``fds''.  It  contains Australia male log mortality rates between 1901 and 2003, provided by The Australian Demographic Data Bank.
\item Berkeley growth: this data set is available in  the  R package ``fda''.  It  contains height measures of 54 girls and 39 boys, under the age of 18, at 31 fixed points.
\end{itemize}
In  \cite{arribas2014}  the authors suggest to smooth the data, since the curves in the first set are very irregular. However, the distance proposed here has an intrinsic smoothing procedure, so we work directly with the original curves. 

We use the proposed Mahalanobis distance to define a depth measure  by   $(1+M_\alpha^2(x,m))^{-1}$, for a realization $x$ of the process. Using this depth, we mark as the functional median the deepest curve of the set. The  central  band of the boxplot is built as the envelope of the 50\% deepest curves, and the  ``whiskers'' are constructed   as the envelop of all the curves that are not marked as outliers. In order to detect the outliers we use the same procedure as before. However, the sample sizes now are too small to estimate robustly the covariance matrix over the grid,  then we use the usual empirical covariance matrix.

The curves marked as outliers for the male mortality set are years 1919 (influenza epidemic) and 1999-2003, which are among the curves detected using other different proposals in \cite{arribas2014}. The resulting boxplot for this data set can be found in Figure \ref{Fig:MortalityBox}, where the outliers are plotted in red. This figure includes also (on the left) a graphic representation of the depths: from green, the deepest curves, to ochre, the outer ones. 

\begin{figure}
    \centering
    \begin{subfigure}[b]{0.47\textwidth}
        \includegraphics[width=\textwidth]{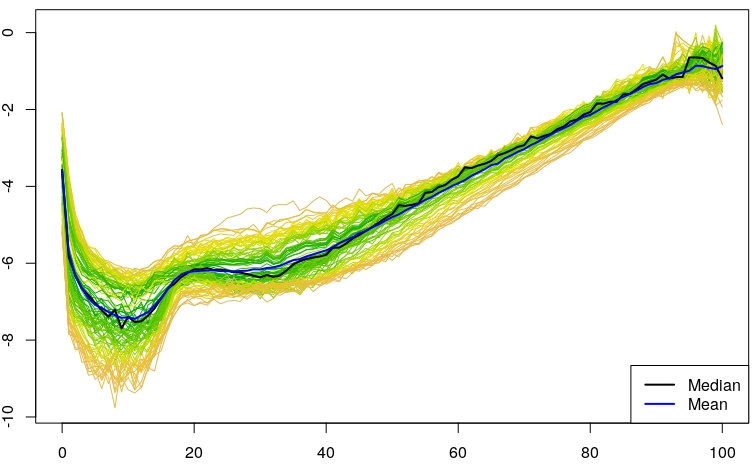}
        \caption{Depths of the curves}
    \end{subfigure}
    \hfill 
    \begin{subfigure}[b]{0.47\textwidth}
        \includegraphics[width=\textwidth]{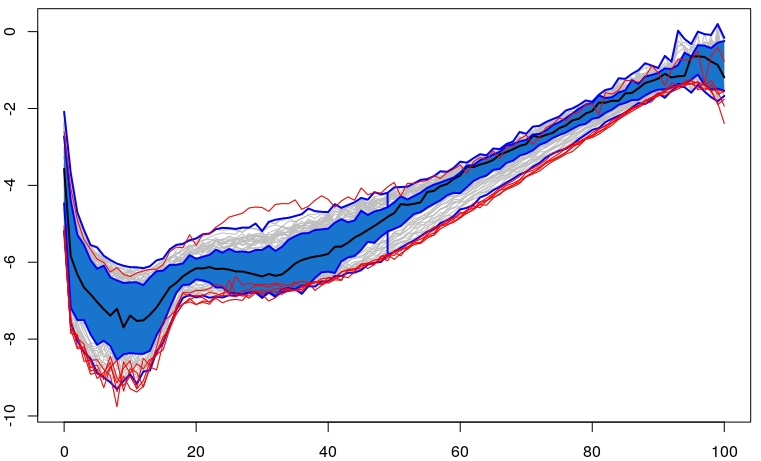}
        \caption{Boxplot and outliers}
        \label{Fig:MortalityBox}
    \end{subfigure}
    \caption{Male mortality rates in Australia 1901-2003}
    \label{Fig:Mortality}
\end{figure}

The boxplots for the Berkely growth sets, female and male, are shown in Figure \ref{Fig:Growth}. The distributions of the distances in this case are far from the theoretical distribution derived for Gaussian processes. In an attempt to overcome this problem, the parameter $\alpha$ is adjusted automatically in order to reduce the Kullback-Leibler divergence between both distributions. The selected values of $\alpha$ with this procedure are $0.089$ for the female set and $0.1$ for the male set. In any case, the number of outliers detected is quite large when compared to the sample size.

Female: 3, 8, 10, 13, 15, 18, 26, 29, 42, 43, 48 and 53.

Male: 5, 10, 15, 27, 29, 32, 35 and 37.

But if we look at the  estimated density functions corresponding to the distribution of  $M^2_\alpha$ on  each set (Figure \ref{Fig:Growth_Density}), we can see that these distributions have two modes. In fact, all the curves marked as outliers are the ones that fall into the second mode (whose distance to the mean is greater that the red dotted line). This behavior is similar to the one of the Integrated Squared Error showed in \cite{arribas2014}.

\begin{figure}
    \centering
    \begin{subfigure}[b]{0.47\textwidth}
        \includegraphics[width=\textwidth]{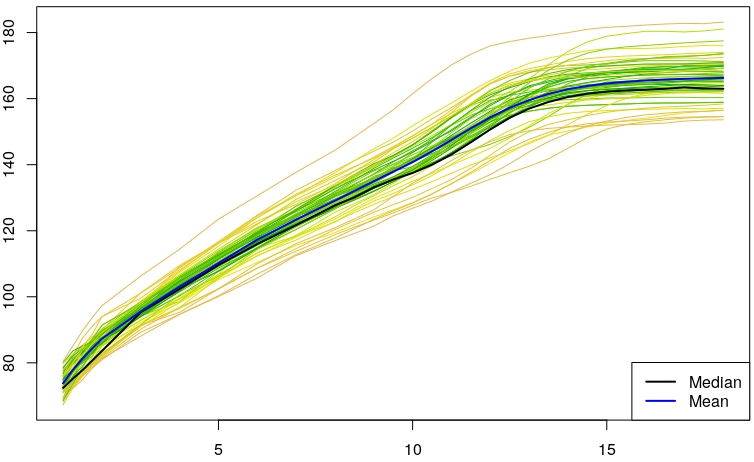}
        \caption{Depths of the curves (female)}
    \end{subfigure}
    \hfill 
    \begin{subfigure}[b]{0.47\textwidth}
        \includegraphics[width=\textwidth]{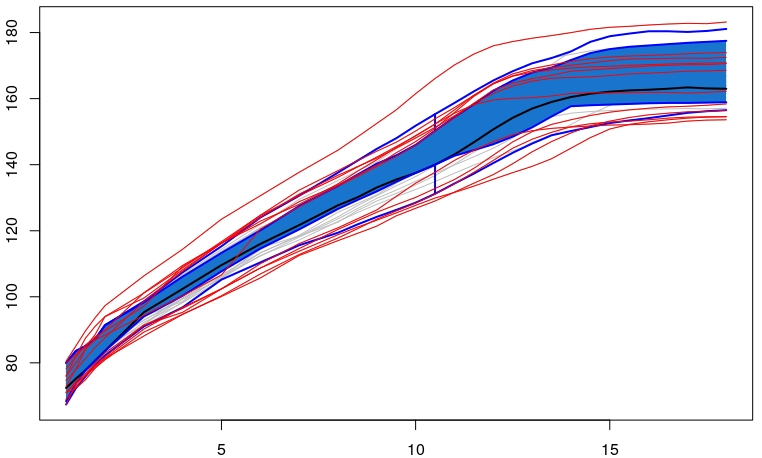}
        \caption{Boxplot and outliers (female)}
    \end{subfigure}
    \begin{subfigure}[b]{0.47\textwidth}
        \includegraphics[width=\textwidth]{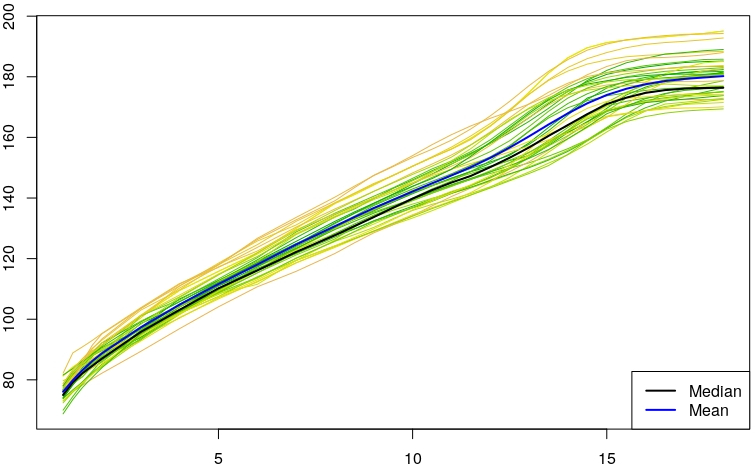}
        \caption{Depths of the curves (male)}
    \end{subfigure}
    \hfill 
    \begin{subfigure}[b]{0.47\textwidth}
        \includegraphics[width=\textwidth]{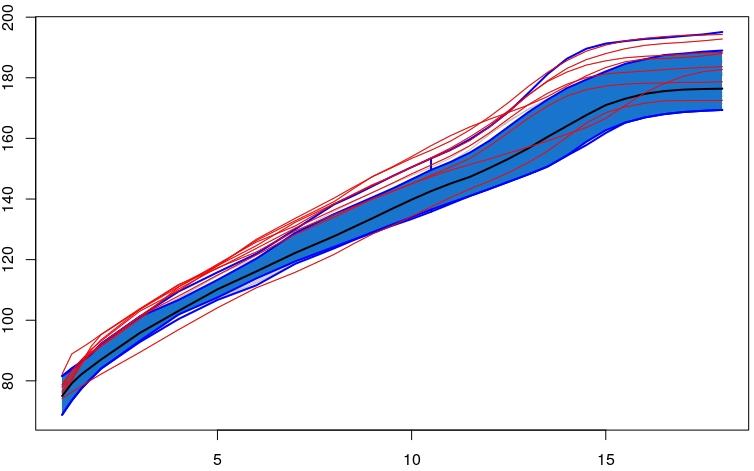}
        \caption{Boxplot and outliers (male)}
    \end{subfigure}
    \caption{Berkeley growth}\label{Fig:Growth}
\end{figure}

\begin{figure}
    \centering
    \begin{subfigure}[b]{0.47\textwidth}
        \includegraphics[width=\textwidth]{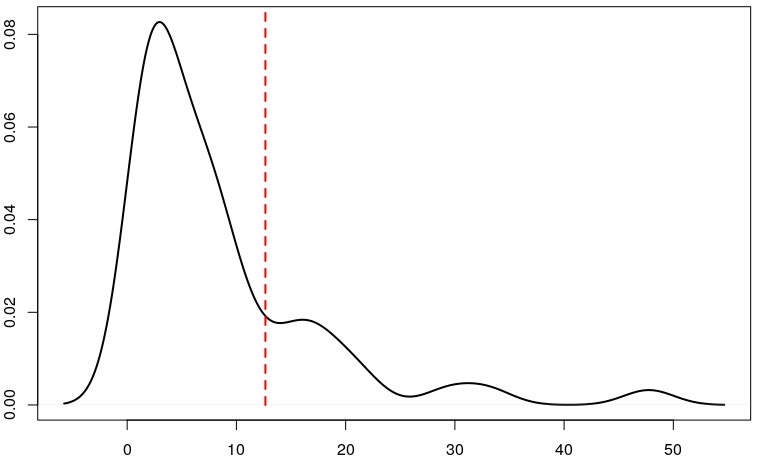}
        \caption{Female}
    \end{subfigure}
    \hfill 
    \begin{subfigure}[b]{0.47\textwidth}
        \includegraphics[width=\textwidth]{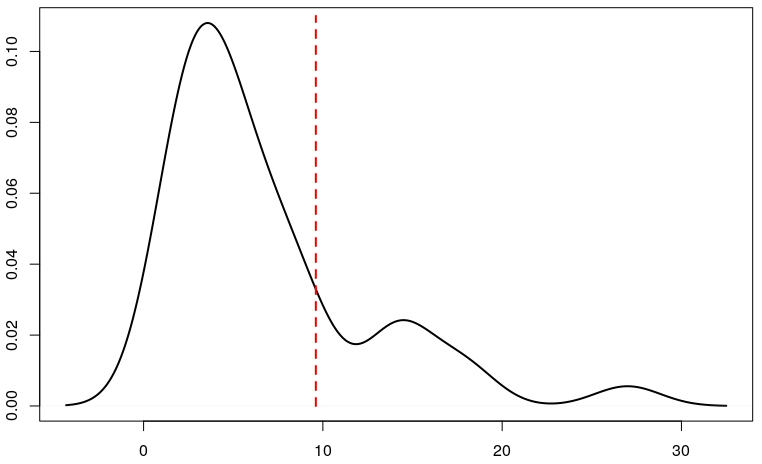}
        \caption{Male}
    \end{subfigure}
    \caption{ Estimated density functions of the distributions of  $M_\alpha^2$   for Berkeley growth.}
	\label{Fig:Growth_Density}
\end{figure}


\subsection{Binary classification}
Mahalanobis distance can be used also for classification, classifying each curve through
the distance to the nearest mean function, whenever the  prior  probabilities $\pi_1,\ldots,\pi_k$ of the classes are equal. When  this  is not the case, the rule to classify a coming observation $x$ is just to assign it to the population $j$ defined by 
$$M_\alpha^2(x,m_j)-2\log \pi_j=\min_{1\leq i\leq k}\left(M_\alpha^2(x,m_i)-2\log \pi_i\right),$$ 
where $m_j$ stands for the mean functions  (for instance, it is used in \cite[Section 3.3]{galeano2015}). Here we present two different  examples  of binary classification with same prior probabilities. In order to check the performance of our proposal, we compare it with other classifiers presented below. The name used on the tables for each method is shown between brackets.

\begin{itemize}
\item Optimal Bayes classifier proposed in \cite{dai2017} (``OB''). This is a functional extension of the classical multivariate Bayes classifier based on  nonparametric estimators of the density functions corresponding to the main coefficients in Karhunen-Lo\`eve expansions. Here the curves are projected onto a common sequence of eigenfunctions, and the previous quotient is taken using the densities of these projections. The authors propose three approaches to estimate these densities. We have chosen the implementation which assumes that these densities are Gaussian since, according to their results, it seems to slightly outperform the others. The number of eigenfunctions used for the projections is fixed by cross-validation.
\item  Mahalanobis-based  semidistance of Equation \eqref{eq:def_galeano}  proposed in \cite{galeano2015} (``$d_{FM}^k$'').
\item k-nearest neighbours with $3$ and $5$ neighbours (``knn3'' and ``knn5''). In spite of its simplicity, this method  tends to show a good performance  when dealing with functional data.
\end{itemize}

Our proposal is denoted as ``$M_\alpha$''. Now the parameter $\alpha$ is fixed by cross-validation, for $\alpha\in [10^{-4},10^{-1}]$. For heterocedastic problems, we have implemented our binary classifier mimicking an improvement that is usually made in the multivariate context. In that finite setting,  given two equiprobable populations with covariance matrices $\Sigma_0, \Sigma_1$,  a curve $x$ is assigned  to class 1, according to the Quadratic Discriminant classifier, whenever 
$$M^2(x,x_0)-M^2(x,x_1) > \log \frac{|\Sigma_1|}{|\Sigma_0|},$$ 
where the finite dimensional Mahalanobis distance $M$ is defined in \eqref{eq:MahMult} (see, for instance, Section 8.3.7 of \cite{izenman2008}). In most cases  with multivariate data, classifying with this rule gives better results  than merely classifying to class 1 when $M^2(x,x_0)>M^2(x,x_1)$.  In the case of functional data this is just   an heuristic improvement. If $m_0, K_0$ and $m_1, K_1$ are the mean and covariance functions of each class, the standard classifier  would assign  the curve $x$ to the class such that $M_{\alpha,K_i}^2(x,m_i)$, $i=0,1$, is minimum  ($M_{\alpha,K_i}$ stands for the distance $M_\alpha$ when using the covariance function $K_i$). Instead,  we will classify $x$ to class 1 if $M_{\alpha,K_0}^2(x,m_0)-M_{\alpha,K_1}^2(x,m_1)>C$, and to 0 if not. This constant $C$ is computed as $\log( (\lambda_1^1\cdot\ldots\cdot\lambda_{10}^1)/(\lambda_1^0\cdot \ldots\cdot\lambda_{10}^0))$, where $\lambda_j^0, \lambda_j^1$, $j=1,\ldots,10$, are the ten greater eigenvalues of $\Gamma_{K_0}$ and $\Gamma_{K_1}$ respectively.

\

\noindent\textit{Cut Brownian Motion and Brownian Bridge}

The first problem under consideration is to distinguish between two ``cut'' versions of a standard Brownian Motion and a Brownian Bridge. By ``cut'' we mean to take the process $X(t)$ on the interval $t\in [0,T]$, $T<1$. We know  an explicit expression for  the Bayes error of this problem, which depends on the cut point $T$. For the case of equal prior probabilities of the classes, which will be the case here,  this Bayes error is  given by,
$$L^* \ = \ \frac{1}{2} - \Phi\left(\frac{(-(1-T)\log(1-T))^{1/2}}{(T(1-T))^{1/2}}\right) + \Phi\left(\frac{(-(1-T)\log(1-T))^{1/2}}{T^{1/2}}\right),$$
where $\Phi$ stands for the distribution function of a standard Gaussian random variable. Since both processes are almost indistinguishable around zero, $L^*\to 0.5$ when $T\to 0$. Also $L^*\to 0$ when $T\to 1$, since then  one can decide the class with no error  just looking at the last point of the curve.

The trajectories of both processes are shown in Figure \ref{Fig:TrayBrowns} and the cut points considered, $ 0.75, 0.8125, 0.875, 0.9375$ and $1$, are marked with vertical dotted lines. For each class, 50 samples are drawn for  training   and 250 for test. The experiment is run 500 times for each cut point, and the trajectories are sampled over an equidistant grid in $[0,1]$ of size $50$.  Table \ref{Table:Bridge} shows  the percentages of misclassified curves, as well as the Bayes errors. Our proposal and knn with 5 neighbors seem to outperform the other methods for this problem.

\begin{figure}
\centering
\includegraphics[scale=0.5]{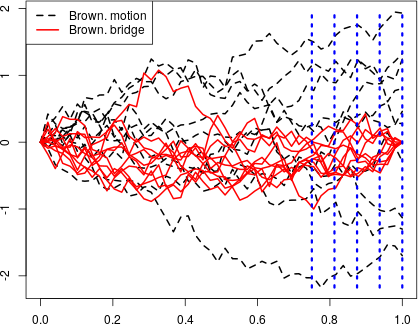}
\caption{Trajectories of Brownian Motion and Bridge with cut points (vertical)}
\label{Fig:TrayBrowns}
\end{figure}

\begin{table}[ht]
\centering
\resizebox{\textwidth}{!}{
\begin{tabular}{lcccccc}
  \hline\hline
t & Bayes & $M_\alpha$ & OB & $d_{FM}^k$ & knn3 & knn5 \\ 
  \hline
0.75 & 33.9 & 42.5  ( 3.5) & 43.5  ( 2.5) & 46.4  ( 3.2) & 43.2  ( 2.8) & \textbf{42.4}  ( 2.8) \\ 
  0.8125 & 30.8 & \textbf{40.0}  ( 3.7) & 41.9  ( 2.6) & 44.8  ( 3.3) & 41.0  ( 2.8) & 40.1  ( 3.0) \\ 
  0.875 & 26.9 & \textbf{36.1}  ( 3.6) & 40.2  ( 2.6) & 42.6  ( 3.7) & 38.0  ( 3.0) & 36.9  ( 3.0) \\ 
  0.9375 & 20.9 & \textbf{32.3}  ( 3.1) & 38.0  ( 2.8) & 39.9  ( 3.5) & 33.7  ( 2.7) & 32.5  ( 2.7) \\ 
  1 & 0.0 & \textbf{26.5}  ( 2.8) & 35.9  ( 2.9) & 36.0  ( 3.5) & 28.4  ( 2.7) & 27.6  ( 2.7) \\ 
   \hline\hline
\end{tabular}}
\caption{Percentage of misclassification for cut Brownian Motion and Brownian Bridge.} 
\label{Table:Bridge}
\end{table}

\

\noindent\textit{Simulated data}

We have implemented also the experimental setting proposed in \cite{dai2017}. The authors consider three different scenarios. For the first two, the curves of both classes $X^{(0)}$ and $X^{(1)}$, are drawn from processes
$$X^{(i)}(t) \ = \ \mu_i(t) + \sum_{j=1}^{50} A_{j,i}\phi_j(t) + \varepsilon, \ \ i=0,1,$$
where $\varepsilon$ is a Gaussian variable with zero mean and variance $0.01$. Function $\phi_j$ is the $j$th element in the Fourier basis, starting with,
$$\phi_1(t) = 1, \ \ \phi_2(t) = \sqrt{2}\cos(2\pi t), \ \ \phi_3(t) = \sqrt{2}\sin(2\pi t).$$
For Scenario A, the coefficients $A_{j,0}, A_{j,1}$ are independent Gaussian variables. For Scenario B they are  independent centered exponential random variables.  Finally, in Scenario C the processes are 
$$X^{(i)}(t) \ = \ \mu_i(t) + \sum_{j=1}^{50} \frac{A_{j,i}}{B_i}\phi_j(t), \ i=0,1,$$
where $A_{j,0},A_{j,1}$ are the same as in Scenario B and $B_0, B_1$ are independent variables with common distribution $\chi_{30}^2/30$. Thus, in this  latter  case the coefficients of the basis expansion are dependent but uncorrelated. The  means  and the variances   of the coefficients $A_{j,i}$, $i=0,1$, are changed in order to check  the ``same'' and ``different'' scenarios for mean and variances.   Then $m_0(t)=0$ always, and  $m_1(t)$ is either 0 or $t$. In the same way, the variance of $A_{j,0}$ is always $\text{exp}(-j/3)$ and the variance of $A_{j,1}$ is  either $\text{exp}(-j/3)$, or $\text{exp}(-j/2)$. The curves are sampled on 51 equidistant points in $[0,1]$.

The prior probabilities of both classes are set to $0.5$ and two sample sizes,  50 and 100, are tested for  training.   For test we use 500 realizations of the processes. Each experiment is repeated 500 times. The misclassification percentages for all the different scenarios are shown in Table \ref{Table:Basis}. Our proposal is mainly the winner, although in Scenario A it is overtaken by the Optimal Bayes classifier   in the case of equal means and different variances. Also knn with 5 neighbors performs better sometimes  in the case of different means and equal variances.

\begin{table}[ht]
\centering
\resizebox{\textwidth}{!}{
\begin{tabular}{lllccccc}
  \hline \hline
  \multicolumn{8}{c}{Scenario A (Gaussian)}\\ 
  n & mean & sd & $M_\alpha$ & OB & $d_{FM}^k$ & knn3 & knn5 \\ 
  \hline
50 & same & diff & 35.9  ( 3.5) & \textbf{19.0}  ( 4.0) & 47.0  ( 3.1) & 45.6  ( 2.2) & 46.2  ( 2.0) \\ 
   & diff & same & 42.3  ( 3.8) & 47.3  ( 6.8) & 43.7  ( 3.7) & 42.9  ( 3.6) & \textbf{42.0}  ( 3.6) \\ 
   & diff & diff & \textbf{29.1}  ( 5.0) & 36.4  ( 10.1) & 40.0  ( 5.4) & 39.7  ( 3.0) & 40.0  ( 3.1) \\ 
  100 & same & diff & 34.2  ( 3.0) & \textbf{9.3}  ( 2.1) & 45.8  ( 3.5) & 44.6  ( 1.9) & 45.4  ( 1.8) \\ 
   & diff & same & \textbf{34.6}  ( 4.5) & 45.1  ( 8.2) & 37.0  ( 4.4) & 42.1  ( 3.0) & 41.0  ( 3.0) \\ 
   & diff & diff & \textbf{22.0}  ( 4.9) & 35.7  ( 11.3) & 34.2  ( 6.2) & 38.3  ( 2.4) & 38.6  ( 2.5) \\ \hline 
  \multicolumn{8}{c}{Scenario B (exponential)}\\ 
  n & mean & sd & $M_\alpha$ & OB & $d_{FM}^k$ & knn3 & knn5 \\ 
  \hline
  50 & same & diff & \textbf{24.2}  ( 5.2) & 30.2  ( 10.4) & 37.0  ( 6.6) & 37.6  ( 2.6) & 38.0  ( 2.7) \\ 
   & diff & same & 41.8  ( 3.9) & 49.1  ( 5.5) & 42.3  ( 4.1) & 38.0  ( 3.4) & \textbf{37.2}  ( 3.6) \\ 
   & diff & diff & \textbf{14.3}  ( 4.8) & 31.8  ( 12.8) & 25.1  ( 9.0) & 24.7  ( 3.1) & 25.1  ( 3.5) \\ 
  100 & same & diff & \textbf{16.9}  ( 3.1) & 24.0  ( 9.6) & 28.2  ( 6.1) & 35.3  ( 2.4) & 35.7  ( 2.3) \\ 
   & diff & same & \textbf{34.5}  ( 4.6) & 48.3  ( 5.9) & 36.7  ( 4.2) & 36.5  ( 2.8) & 35.6  ( 2.7) \\ 
   & diff & diff & \textbf{7.7}  ( 2.9) & 30.1  ( 13.4) & 17.8  ( 6.3) & 21.6  ( 2.4) & 21.8  ( 2.6) \\ \hline 
  \multicolumn{8}{c}{Scenario C (dependent)}\\ 
  n & mean & sd & $M_\alpha$ & OB & $d_{FM}^k$ & knn3 & knn5 \\ 
  \hline
  50 & same & diff & \textbf{30.0}  ( 5.4) & 33.3  ( 8.1) & 40.1  ( 5.9) & 39.9  ( 2.7) & 39.9  ( 2.7) \\ 
   & diff & same & 43.6  ( 4.1) & 48.8  ( 4.8) & 42.9  ( 4.2) & 38.1  ( 3.6) & \textbf{37.5}  ( 3.8) \\ 
   & diff & diff & \textbf{19.9}  ( 4.9) & 36.2  ( 11.0) & 30.3  ( 7.7) & 26.4  ( 3.1) & 26.6  ( 3.3) \\ 
  100 & same & diff & \textbf{21.7}  ( 3.0) & 28.0  ( 7.5) & 29.4  ( 5.7) & 37.6  ( 2.4) & 37.5  ( 2.4) \\ 
   & diff & same & 38.0  ( 4.3) & 48.8  ( 5.0) & 38.9  ( 3.8) & 36.5  ( 2.7) & \textbf{35.6}  ( 2.8) \\ 
   & diff & diff & \textbf{13.3}  ( 3.2) & 34.6  ( 11.0) & 23.2  ( 6.1) & 23.4  ( 2.4) & 23.3  ( 2.4) \\ 
   \hline \hline
\end{tabular}}
\caption{Percentage of misclassification for the experimental setting of \cite{dai2017}.} 
\label{Table:Basis}
\end{table}

\section*{Acknowledgements}
This work has been partially supported by Spanish Grant MTM2016-78751-P. The authors are very grateful to Daniel Est\'evez for his valuable help with operator theory.



\bibliography{Refs}{}
\end{document}